\documentclass{fundam}

\usepackage{amssymb}
\usepackage[hmargin=2cm,vmargin=3cm]{geometry}
\usepackage{graphicx,url}
\usepackage{url,hyperref,color}
\newcommand{\crishref}[2] {\href{#1}{\color{blue}{#2}}}

\newtheorem{thm}{Theorem}
\newtheorem{cor}[thm]{Corollary}

\newtheorem{prop}[thm]{Proposition}

\newtheorem{defn}[thm]{Definition}

\DeclareMathOperator{\dom}{dom}

\newcommand{\Z}{\mathbb{Z}^{+}}

\begin{document}
\setcounter{page}{1001}
\issue{XXI~(2014)}

\sloppy

\title{Universality and Almost Decidability}
\author{Cristian S. Calude$^{1}$, Damien Desfontaines$^{2}$\\[2ex]
$^{1}$Department of Computer Science\\
University of Auckland,
Private Bag 92019, Auckland, New Zealand\\
\url{www.cs.auckland.ac.nz/~cristian}\\
$^{2}$\'{E}cole Normale Sup\'erieure,
45 rue d'Ulm, 75005 Paris, France\\
\url{desfontain.es/serious.html}
}
%\date{\today}

\maketitle

\runninghead{C. S. Calude, D. Desfontaines}{Universality and Almost Decidability}

\begin{abstract}
We present and study new definitions of universal and programmable universal unary functions and consider a new simplicity criterion:
almost decidability of the halting set.
 A  set of positive integers $S$ is almost decidable if  there exists a  decidable and generic (i.e.\ a  set of natural density one)      set  whose intersection with $S$ is decidable. 
Every decidable set is almost decidable, but the converse implication is false.   We prove the existence of infinitely many universal functions whose halting sets are 
generic (negligible, i.e.\  have density zero) and (not) almost decidable.  One result---namely, the existence of  infinitely many universal functions whose halting sets are 
generic (negligible) and not almost decidable---solves  an open problem in \cite{HM}. We conclude with some open problems.
\end{abstract}
 

 \section{Universal Turing Machines and Functions}
 
 The first universal Turing machine was constructed by Turing  \cite{AT,ATc}. In Turing's words:
\begin{quote}\it
%It can be shown that
\dots a single special machine of that type can be made to do the work of all. It could in fact be made to work as a model of any other machine. The special machine may be called the universal machine.
\end{quote}

Shannon \cite{CS} proved that two symbols were sufficient for constructing a universal Turing machine providing enough states can be used.  According to Margenstern~\cite{Maurice-2010}: ``Claude Shannon raised the problem of what is now called the {\em descriptional complexity} of Turing machines: how many states and letters are needed in order to get universal machines?''
Notable universal Turing machines  include the machines constructed by
 Minsky    (7-state 4-symbol) \cite{MM},   Rogozhin (4-state 6-symbol) \cite{YR}, 
 Neary--Woods (5-state 5-symbol) \cite{NW}.
Herken's book \cite{Herken} celebrates the first 50 years of universality.  Woods and  Neary  presents a  survey  in \cite{WN}; 
Margenstern's paper ~\cite[p.\ 30--31]{Maurice-2010} presents also a time line of the main results. 

Roughly speaking, a universal   machine is a  machine capable of simulating any other   machine. There are a few  definitions of universality, the most important being {\em universality in Turing's sense} and {\em programmable universality} in the sense of Algorithmic Information Theory~\cite{Ca,DH}.

In the following we denote by $\mathbb{Z}^{+}$ the set of positive integers $\left\{ 1,2,\ldots\right\} $,
and $\overline{\mathbb{Z}^{+}}=\mathbb{Z}^{+}\cup\left\{ \infty\right\}$. The cardinality of a set $S$ is denoted by $\#S$. The domain of a partial function $F\colon \Z \longrightarrow\overline{\Z}$ is $\dom(F)=\{x \in \mathbb{Z}^{+}\mid F(x)\not=\infty\}$.
We assume familiarity with the basics of computability theory~\cite{Cooper-2004,YM2010}.

\medskip

We define now universality for unary functions.

A partially computable function $U \colon\mathbb{Z}^{+}\longrightarrow\overline{\mathbb{Z}^{+}}$ is called {\em (Turing) universal} if there exists a  computable function $C_U \colon\mathbb{Z}^{+}\times\mathbb{Z}^{+}\longrightarrow\mathbb{Z}^{+}$ such that for any partially computable function $F\colon\mathbb{Z}^{+}\longrightarrow\overline{\mathbb{Z}^{+}}$ there exists an integer $g_{U,F}$ (called a {\em G\"{o}del number} of $F$ for $U$) such that  for all $ x\in\mathbb{Z}^{+}$ we have: $U \left(C_U\left(g_{U,F},x\right)\right)=F\left(x\right)$. 

Following  \cite{YM2012,CD} we say that a partially computable function $U \colon\mathbb{Z}^{+}\longrightarrow\overline{\mathbb{Z}^{+}}$ is  {\em programmable  universal} if for every partially computable function $F\colon\mathbb{Z}^{+}\longrightarrow\overline{\mathbb{Z}^{+}}$ there exists a constant $k_{U,F}$ such that for every $x\in\Z$  there exists $y\le k_{U,F}\cdot x$ with $U (y) = F (x).$\footnote{For  the programming-oriented reader we note that the property  ``programmable  universal" corresponds to being able to write a compiler.}

\begin{thm}
\label{mainprogramuniv}
A partially computable function $U \colon\mathbb{Z}^{+}\longrightarrow\overline{\mathbb{Z}^{+}}$ is  programmable  universal iff there exists a partially computable function $C_U \colon\mathbb{Z}^{+}\times\mathbb{Z}^{+}\longrightarrow\overline{\mathbb{Z}^{+}}$ such that for any partially computable function $F\colon\mathbb{Z}^{+}\longrightarrow\overline{\mathbb{Z}^{+}}$ there exist two integers $g_{U,F}, c_{U,F}$ such that  for all $ x\in\mathbb{Z}^{+}$ we have 
\begin{equation}
\label{suniv1}
U \left(C_U\left(g_{U,F},x\right)\right)=F\left(x\right)
\end{equation}
and 
\begin{equation}
\label{suniv2}C_U\left(g_{U,F},x\right)\leq c_{U,F}\cdot x.
\end{equation}
\end{thm}
\begin{proof}First we construct a partially computable function $V\colon\mathbb{Z}^{+}\longrightarrow\overline{\mathbb{Z}^{+}}$ and a partially computable function $C_V \colon\mathbb{Z}^{+}\times\mathbb{Z}^{+}\longrightarrow\overline{\mathbb{Z}^{+}}$   such that  for every partially computable function $F$,  (\ref{suniv1}) and (\ref{suniv2})  are satisfied.  Indeed, the classical Enumeration Theorem~\cite{Cooper-2004} shows the existence of a partial 
computable function $\Gamma \colon\mathbb{Z}^{+}\times \Z \longrightarrow\overline{\mathbb{Z}^{+}}$
 such that for every
 partial computable function $F\colon \Z \longrightarrow\overline{\Z}$ there exists $e \in \mathbb{Z}^{+}$
such that $F(x) = \Gamma(e, x)$, for all $x\in \mathbb{Z}^{+}$. Consider the computable function
$f \colon \Z\times \Z \longrightarrow\Z$  such that the binary expansion of $f(e, x)$
 is obtained by prefixing the binary expansion of $x$ with the binary expansion of $2e+1$. Then 
$\alpha$ is injective because if $e_{1}e_{2}\dots e_{n}$ and $x_{1}x_{2}\dots x_{m}$ are the binary expansions of $e$ and $x$, respectively, then $e_{1}0e_{2}0\dots e_{n}1x_{1}x_{2}\dots x_{m}$ is the binary expansion of $f(e,x)$ from which we can uniquely recover $e$ and $x$. If $f_1, f_2 \colon \Z \longrightarrow \Z$ are  computable 
partial inverses of $f$, i.e.\ $ f(f_1(x), f_2(x)) = x$, for all $x\in f(\Z\times\Z)$,
 then the function $V(x) = \Gamma(f_1(x), f_2(x))$ has  (\ref{suniv1}) and (\ref{suniv2}) for $C_{V}=f$.\footnote{This construction
 suggests that the function $C_{U}$ may be taken to be computable.}

Let $U$ be programmable  universal, that is, for every partially computable function $F\colon\mathbb{Z}^{+}\longrightarrow\overline{\mathbb{Z}^{+}}$ there exists a constant $k_{U,F}$ such that for every $x\in\Z$  there exists  $y\le k_{U,F}\cdot x$ with $U (y) = F (x).$  We shall use $V$ to prove that $U$ satisfies the condition in the statement of the theorem.

Let $b\colon \Z\times\Z \longrightarrow\Z$ be a computable bijection and $b_1, b_2$ the components of its inverse. 

We define the partially computable function $C_U$ as follows.  We consider first the set $ S(z,x)=\{y \in\dom(U)\mid y \le b_1(z)\cdot x, U(y)=V(C_V(b_2(z),x))\}$ and then we define  $C_U(z,x) $ to  be the first element of $S(z,x)$ according to some computable enumeration of $\dom(U)$. Formally, let $E$ be a computable one-one enumeration of $\dom(U)$ and  define

\[C_U(z,x) = E\left(\inf\{y\:|\:E(y) \le b_1(z)\cdot x \text{ and }U(E(y))=V(C_V(b_2(z),x))\}\right).\]

We now prove that $U$ satisfies the condition in the statement of the theorem via $C_U$. To this aim let $F$ be a partially computable function
 and let $g_{V,F}, c_{V,F}$ be the constants associated to $V$ and $F$. \\[-2ex]

Put   $g_{U,F} = b(k_{U,F}, g_{V,F})$ and $c_{U,F}= k_{U,F}$.  \\[-2ex]

We have:
\begin{eqnarray*}
    C_{U}(g_{U,F},x) & = & E\left(\inf\{y\:|\: E(y) \le b_1(g_{U,F})\cdot x
\text{ and }U(E(y))=V(C_V(b_2(g_{U,F}),x))\}\right)\\
                     & = & E\left(\inf\{y\:|\: E(y)\le k_{U,F}\cdot x \text{ and } U(E(y))=V(C_V(g_{V,F},x))\}\right)\\
                     & = & E\left(\inf\{y\:|\: E(y) \le k_{U,F}\cdot x \text{ and } U(E(y))=F(x)\}\right)\\
                     & \le & k_{U,F}\cdot x=c_{U,F}\cdot x,
\end{eqnarray*}
and $U(C_{U}(g_{U,F},x))=F(x)$.

\if01
We have:
\begin{eqnarray*}
    C_{U}(g_{U,F},x) & = & E(\inf\{E(y) \le b_1(g_{U,F})\cdot x \text{ and } U(E(y))=V(C_V(b_2(g_{U,F}),x))\}\\
                     & = & E(\inf\{y \le k_{U,F}\cdot x \text{ and } U(E(y))=V(C_V(k_{U,F},x))\})\\
                     & = & E(\inf\{y \le k_{U,F}\cdot x \text{ and } U(E(y))=F(x)\}),\\
& \le & k_{U,F}=c_{U,F},
\end{eqnarray*}
and $U(C_{U}(g_{U,F},x))=F(x)$.
\fi

Conversely, if $V$ satisfies (\ref{suniv1}) and (\ref{suniv2})  with the partially computable function $C_{V}$, then $V$ is programmable  universal: given  a partially computable function $F$ and $x\in\Z$, $y=C_{V}(g_{V,F},x)$ and $k_{V,F}=c_{V,F}$.
\end{proof}

Universal and programmable  universal functions exist and can be effectively constructed. Every programmable  universal function is  universal, but the converse implication is false.

\section{The Halting Set and Almost Decidability}
Interesting classes of   Turing machines have decidable halting sets: for example, Turing machines with two letters and two states~\cite{Maurice-2010}. In contrast, the
 most (in)famous result in computability theory is that {\em the halting set  $\text{Halt}(U)=\dom(U)$ of  a universal function $U$ is undecidable.}

However, the halting set
$\text{Halt}(U)$ is  computably enumerable (see~\cite{Cooper-2004,YM2010}).  
How ``undecidable'' is  $\text{Halt}(U)$?   To answer this question we formalise the following notion: a set $S$ is ``almost decidable'' if there exists a ``large'' decidable set  whose intersection with  $S$ is also decidable. In other words, the undecidability of $S$ can be located to a ``small'' set. 

To define ``large''  sets we can employ measure theoretical  or topological tools adapted to the set of positive integers
(see
\cite{Ca}). 
In what follows we will
work with the {\it (natural) density}  on $\mathcal{P}\left(\mathbb{Z}^{+}\right)$. Its motivation is the following.  If a positive integer is ``randomly'' selected from the set
$\{1,2,\dots ,N\}$, then the probability that it belongs to a given set $A \subset \mathbb{Z}^{+}$ is 

\[
p_{N}\left(A\right)=\frac{\#\left(\left\{ 1,\ldots,N\right\} \cap A\right)}{N}\raisebox{.5ex}.
\]

If $\lim_{N\longrightarrow\infty}p_{N}\left(A\right)$ exists, then 
the set $A\subset \mathbb{Z}^{+}$  has {\em density}:

\[
d\left(A\right)=\lim_{N\longrightarrow\infty}\frac{\#\left\{1\leq  x\leq N\:|\: x\in A\right\} }{N}\raisebox{.5ex}.
\]

%In a sense, density $d (A)$ models  ``the probability that a randomly chosen integer $x\in \mathbb{Z}^{+}$ is in $A$''. 

\begin{defn}A set is {\em generic} if it has density one; a set of density zero is called {\em negligible}.  A set $S\subset \Z$ is {\em almost decidable} if  there exists a generic decidable set $R\subset \Z$  such that $R \cap S$ is decidable.
\end{defn}

Every decidable set is almost decidable, but, as we shall see below, there exist almost decidable sets which are not decidable. A  set which is not almost decidable contains no generic decidable subset; of course, this result is non-trivial if the set itself is generic.

\begin{thm}[\cite{HM}, Theorem 1.1] 
\label{hautm} There exists a universal Turing machine  whose halting set is negligible and almost  decidable (in polynomial time).
\end{thm}

A single semi-infinite tape, single halt state, binary alphabet universal Turing machine satisfies Theorem~\ref{hautm}; other examples are provided in \cite{HM}.

Negligibility  reduces to some extent the power of almost decidability in Theorem~\ref{hautm}.
This deficiency is overcome in the next result: the price  paid is in the redundancy of the universal function.

\begin{prop}
\label{genericnad}There exist infinitely many  universal functions  whose halting sets are generic and almost decidable (in polynomial time).
\end{prop}

\begin{proof}Let $V$ be a universal function and define $U$ by the formula:
\[U(x)  = \left\{ \begin{array}{ll}
V(y), & \mbox{\rm if $x=y^{2}$, for some $y\in\Z$}, \\
0, & \mbox{\rm otherwise} \,.
  \end{array} \right.\]
Clearly, $U$ is universal, $\text{Halt}(U)$ is generic, the set $S=\{ y \in\Z\mid 
y\not= x^{2} \text  {   for every } x\in\Z\}$ is
generic and decidable (in polynomial time) and $S\cap \text{Halt}(U)$ is generic and decidable (in polynomial time).
\end{proof}

\begin{cor} There exist  infinitely many almost decidable but not decidable sets.
\end{cor}

Does there exist  a universal function $U$ whose halting set  is not almost  decidable? This problem was left open in \cite{HM}: here we answer it in the affirmative.

\begin{thm}
\label{notalmostdecid}
There exist infinitely many universal functions  whose halting sets are not negligible and not almost decidable.
\end{thm}

\begin{proof}
We start with an arbitrary  universal function   $V$ and construct a new  universal function $U$ whose halting set $\text{Halt}(U)$ is not almost decidable.\\

First we  define the computable function $\varphi\colon \mathbb{Z}^{+}\longrightarrow\mathbb{Z}^{+}$
by $
\varphi (n)=\max\{ k\in\Z \mid 2^{k-1} \mbox{ divides } n\}.$
\if01
Here are the
first values of $\varphi$:
\bigskip

\noindent \begin{center}
\begin{tabular}{|c|c|c|c|c|c|c|c|c|c|c|}
\hline 
$x$ & 1 & 2 & 3 & 4 & 5 & 6 & 7 & 8 & 9 & 10  \tabularnewline 
\hline 
$\varphi\left(x\right)$  & 1 & 2 & 1 & 3 & 1 & 2 & 1 & 4 & 1 & 2 \tabularnewline
\hline 
\hline
$x$ & 11 & 12 & 13 & 14 & 15 & 16 & 17 & 18 & 19 & 20\tabularnewline 
\hline 
$\varphi\left(x\right)$  & 1 & 3 & 1 & 2 & 1 & 5 & 1 & 2 & 1 & 3\tabularnewline
\hline 
\end{tabular}
\par\end{center}
\fi
\bigskip 

The function $\varphi$ has the following properties:\\[-2ex]

\begin{enumerate}
\item[(a)] $\varphi(2^{m-1}(2k+1))=m$, for every $m,k\in\Z$, so  $\varphi$ outputs every positive integer infinitely many times.
\item[(b)] $\varphi^{-1}(n) = \{ k\in\Z\mid 2^{n-1} \text { divides } k \text{ but } 2^{n}
\text{ does not divide } k\}$.
\item[(c)] $d(\varphi^{-1}(n))=2^{-n}$, for all $n\in\Z$.
\item[(d)] If $S\subseteq
\Z$  and $d(S)=1$, then for every $n\in\mathbb{Z}^{+}$,
$\varphi^{-1}\left(n\right)\cap S\neq\emptyset$.\\[-2ex]
\end{enumerate}

For (d) we note that  if $\varphi^{-1}\left(n\right)\cap S=\emptyset$, then $d\left(S\right)\leq1-2^{-n}$, a contradiction. 

\medskip

%is the set of positive integers $k$ such  that $2^{n-1}$ is a divisor of $k$, but not $2^{n}$.

Next we  define $U(x)=V(\varphi (x))$ and prove that $U$ is  universal. We consider the partially computable function
$C_{U}(z,x) = \inf\{s\in \Z \mid \varphi(s) = C_{V}(z,x)\}$ and note that: 1) by (a), $\dom (C_{U}) = \dom (C_{V})$, and 2) 
$\varphi (C_{U}(z,x))= C_{V}(z,x)$, for all $(z,x)\in  \dom (C_{V})$.  Consequently, for every partially computable function
$F\colon\mathbb{Z}^{+}\longrightarrow\overline{\mathbb{Z}^{+}}$ we have $F(x) = V(C_{V}(g_{V,F},x)) =
V(\varphi(C_{U}(g_{V,F},x)))$, so $g_{U,F}=g_{V,F}$.

\medskip

Let us assume by absurdity that there exists a generic decidable set $S\subseteq \Z $ such that $S\cap \text{Halt}(U)$ is decidable.

\medskip

Define the partial function $\theta \colon \mathbb{Z}^{+}\longrightarrow\overline{\mathbb{Z}^{+}}$ by
$\theta(n) = \inf\{k\in S\mid \varphi\left(k\right)=n\}.$
\medskip

 As $S$ is decidable,  $\theta$ is partially computable;  by (a)  ($\varphi$ is surjective) and by (d) (as $d(S)=1$, for all $n\in\Z$, $
\varphi^{-1}(n) \cap S \not=\emptyset$)
it follows that $\theta$ is computable. Furthermore, the computable function $\theta$ has the following two properties:
 for all $n\in\Z$, $\varphi(\theta(n))=n$ and $\theta(n)\in S$.
  
 \medskip
We next prove that for all $n\in\Z$,

\begin{equation}
\label{hequiv}
n \in \text{Halt}(V) \, \text{   iff } \, \theta(n) \in  S\cap \text{Halt}(U).
\end{equation}

Indeed,
\begin{eqnarray*}
n \in \text{Halt}(V) &\Longleftrightarrow &  V(n)<\infty\\
                     &\Longleftrightarrow  & V(\varphi(\theta(n)))<\infty \,\,\,\,\,\,\,\,\,\,\,\,\,\,\,\,\,\,\,\,\,\,\,\,\,\,\,\,(\varphi(\theta(n)) = n)\\
                     &\Longleftrightarrow  & U(\theta (n))<\infty\, \, \,\,\,\,\,\,\,\,\,\,\,\,\,\,\,\,\,\,\,\,\,\,\,\,\,\,\,\,\,\,\,\,\, (\text{definition of  } U)\\
                     &\Longleftrightarrow  & \theta (n) \in \text{Halt}(U)\\
                    & \Longleftrightarrow  & \theta (n) \in S\cap \text{Halt}(U). \,\,\,\,\,\,\,\,\,\,\,\,\,\,\,\,\,\, (\theta (n) \in S)
\end{eqnarray*}

From  (\ref{hequiv}) it follows that  $\text{Halt}(V)$ is decidable because $S\cap \text{Halt}(U)$ is decidable, a contradiction.

\medskip

Finally, $d(\text{Halt}(U))>0$ because $\text{Halt}(U)=
\varphi^{-1}(\text{Halt}(V))$.

 By varying the universal function $V$ we get infinitely many examples of universal functions $U$.
\end{proof}

\begin{cor}There exist infinitely many  universal  functions $U$ such that for any generic computably
enumerable set $S\subseteq\mathbb{Z}^{+}$, $S\cap \text{Halt}\left(U\right)$
is not decidable.
\end{cor}
\begin{proof}
Assume $S$ is  computable enumerable and $d(S)=1$. If   replace
the computable function $\theta$ with the computable function $\Gamma (n) = E(\min\{k\in\Z\mid \varphi(E(i))=n\})$, where
$E \colon \mathbb{Z}^{+}\longrightarrow\overline{\mathbb{Z}^{+}}$ is a computable injective function such that
$E(\Z)=S$ ($S$ is infinite) in the proof of Theorem~\ref{notalmostdecid}, then  we prove that $S\cap\text{Halt}\left(U\right)$
is not decidable.
\end{proof}

There are six possible relations between the notions of negligible, generic and  almost  decidable sets.
The above results looked at three of them: here we show that the remaining three possibilities can be realised too.
First, it is clear that there exist non-negligible and decidable sets, hence non-negligible and almost decidable sets.

\medskip

The next result is a stronger form of Theorem~\ref{notalmostdecid}: its proof depends on a set $A$ and works
for other interesting sets as well.\medskip

\begin{thm}
\label{gnad}
There exist infinitely many universal functions whose halting sets are 
generic and not almost decidable.
\end{thm}
\begin{proof} We 
use a computably enumerable generic set $A$ which has no generic
decidable subset (see Theorem~2.22 in \cite{JS2012}) to construct a  universal function as in the statement above.

Assume  $A=  \text{Halt}(F)$ for some  partially computable function $F$. Let $V$ be
an arbitrary  universal function and  define $U$ by:

\[U(x)  = \left\{ \begin{array}{ll}
V(y), & \mbox{\rm if $x=y^{2}$, for some $y\in\Z$}, \\
F(x), & \mbox{\rm otherwise} \,.
  \end{array} \right.\]
  
 Clearly  $\text{Halt}(U)$ is universal and  generic.

 For the sake of a contradiction assume that $\text{Halt}(U)$ is almost decidable by $S$, i.e.\ $S$ is a generic
 decidable set such that $\text{Halt}(U) \cap S$ is decidable.
 
 We now prove that $\text{Halt}(F)$ is almost decidable by $S'=S\cap \overline{P}$, where $P$ is the set of square positive
 integers (note that $P$ is decidable and negligible) and $\overline{P}$ is the complement of $P$. It is clear that $S'$ is generic and decidable,  so we need only to show that
 $\text{Halt(F)} \cap S^{'} = \text{Halt(F)} \cap  S \cap \overline{P}$ is decidable. 
 
 We note that $\text{Halt}(U)$ is a disjoint union of the sets  $\{x\in\Z \mid x=y^{2}, \text{ for some } y\in \text{Halt}(V)\}$ and
$\text{Halt}(F) \cap  \overline{P}$, and the first set is a subset of $P$. To test whether $x$ is in $\text{Halt(F)} \cap  S^{'}$ we proceed as follows:
a) if $x\in P$, then $x\not\in \text{Halt}(F) \cap S^{'}$, b) if $x\not \in P$, then $x\in \text{Halt(F)} \cap  S^{'}$ iff $x\in \text{Halt}(U)\cap S$.
Hence, $\text{Halt(F)} \cap S^{'}$ is decidable because $\text{Halt}(U) \cap S$ is decidable, so $\text{Halt}(U)$ is almost
decidable.
 
 We have obtained a contradiction because $\text{Halt(F)} \cap S^{'}$ is a generic decidable subset of $A$, hence $\text{Halt}(U)$ is not almost decidable.
\end{proof}

%The proof of Theorem~\ref{gnad}  can be adapted to different sets $A$.\medskip

Let $r \in (0,1]$. We say that a set $S\subset \Z$ is {\em $r$-decidable}  if  there exists a  decidable set  $R\subset \Z$  such that $d(R)=r$ and  $R \cap S$ is decidable;  a set $S\subset \Z$ is {\em weakly decidable} if  $S$ is  $r$-decidable for some $r \in (0,1)$.  With this terminology, generic sets coincide with  $1$-decidable sets.

\medskip

Theorem
3.18 of \cite{DJS2012}  states that there is a computably enumerable  generic set that has
no decidable subset of density in  $(0,1)$.
Using this set in the proof of Theorem~\ref{gnad} 
we get the following stronger result: 

\begin{thm}There exist infinitely many universal functions whose halting sets are 
generic and not weakly decidable.
\end{thm}

  A simple set is a  co-infinite computably enumerable   set  whose complement 
 has no decidable subset; the existence of a negligible simple set is shown in the proof of Proposition 2.15 in \cite{JS2012}.
If in the proof of Theorem~\ref{gnad} we use  a negligible simple set instead of the computably enumerable  generic set which has no generic
decidable subset we obtain the following result:

\begin{thm}There exist infinitely many universal functions whose halting sets are 
negligible and not almost decidable.
\end{thm}

\section{A Simplicity Criterion for Universal Functions and Open Problems}

Universality is one of the most important concepts in computability theory.  However, not all universal machines are made equal. 
The most popular criterion for distinguishing between universal Turing machines   is  the number of states/symbols. Other three other criteria
 of simplicity for universal prefix-free Turing machines have been studied in \cite{cris-2010}. The property of almost decidability is another criterion of simplicity for  universal functions.

\medskip

The universal function $U$ constructed in the proof of Theorem~\ref{notalmostdecid} is {\em not} programmable  universal.
Theorems 2 and 8 in \cite{CNSS} show that the halting sets of programmable universal string functions
(plain or prefix-free) are never negligible. {\em Are there programmable  universal functions not almost decidable?}

\medskip

 The notion of almost decidability suggests the possibility of an approximate (probabilistic) solution for the halting problem
(see also \cite{CS2008,CD}). Assume that the halting set is $\text{Halt}(U)$ is almost decidable via the generic decidable set $S$
and we wish to test whether an arbitrary $x\in\Z$ is in $\text{Halt}(U)$. If $x\in S$, then $x\in \text{Halt}(U)$ iff $x\in S \cap \text{Halt}(U)$.
If $x\not\in S$, then we don't know whether $x\in \text{Halt}(U)$ or $x\not\in \text{Halt}(U)$ (the undecidability is located in $ \overline{S}\cap \text{Halt}(U)$). Should we conclude that $x\in \text{Halt}(U)$ or $x\not\in \text{Halt}(U)$?  Density does not help because
$d(\overline{S} \cap \text{Halt}(U))= d(\overline{S} \cap \overline{\text{Halt}(U)})=0$. 
%If $d(\text{Halt}(U))> d(\overline{\text{Halt}(U)})$
%(for example, if $\text{Halt}(U)$ is generic) we may be inclined to conclude that $x\in \text{Halt}(U)$, but but it is still possible that $\overline{S} \cap \text{Halt}(U)$ is a negligible subset of $ \overline{S}$.
 It is an open problem
to find a solution.

%  {\em Does there exist a universal function $U$ such that $\text{Halt}(U)$ is generic and not almost decidable?}
%For such a $U$,  $\text{Halt}(U)$ contains no generic decidable subset, a property which is an analogue
%of immunity for computably enumerable sets. 

\medskip 

 The notion of almost decidability  can be refined in different ways, e.g.\ by looking at the computational complexity of the decidable sets appearing in Theorem~\ref{notalmostdecid}. Also, it will be interesting to study the property of {\em almost decidability} topologically or for other densities.
 
\section*{Acknowledgement} The authors have been supported in part by the Quantum Computing Research Initiatives at Lockheed Martin.
We thank  Laurent Bienvenu, Ludwig Staiger, Mike Stay  and  Declan Thompson for interesting discussions on the topic of this paper.
We also thank the anonymous referees for comments which improved the paper.


\begin{thebibliography}{99}
 \bibitem{Ca} C. S. Calude. {\em Information and Randomness. An Algorithmic
    Perspective}, 2nd Edition, Revised and Extended, Springer Verlag, Berlin,
  2002.
 
\bibitem{cris-2010} C. S. Calude. Simplicity via provability for universal prefix-free Turing machines,
{\em Theoretical Comput. Sci.} 412 (2010), 178--182. 

\bibitem{CD} C. S. Calude, D. Desfontaines. Anytime Algorithms for Non-Ending Computations, 
{\em  CDMTCS Research Report} 463, 2014.

\bibitem{CNSS} C. S. Calude, A. Nies, L. Staiger,  F. Stephan. Universal recursively enumerable sets of strings, {\em Theoretical Comput. Sci.}  412 (2011), 2253--2261.


\bibitem{CS2008} C. S. Calude, M. A. Stay. Most 
programs stop quickly or never halt, {\em Advances in Applied Mathematics}, 40 (2008), 295--308.

\bibitem{Cooper-2004} S.  B. Cooper. {\em Computability Theory}, Chapman \& Hall/CRC London, 2004.
%\bibitem{Cook} M. Cook. Universality in elementary cellular automata,
%{\em Complex Systems} 15(1) (2004), 1--40. 

 \bibitem{DH} R. Downey, D. Hirschfeldt. {\em Algorithmic Randomness
    and Complexity}, Springer, Heidelberg,  2010.
    
    
    \bibitem{DJS2012}  R.  G.   Downey, C.  G. Jockusch Jr.,  P.  E.  Schupp.
    \crishref{http://www.newton.ac.uk/preprints/NI12039.pdf}{Asymptotic Density and Computably Enumerable Sets},
    Newton Institute Preprint     NI12039, 2012.
    
\bibitem{HM} J. D. Hamkins, A. Miasnikov. The halting problem is decidable on a set of asymptotic probability one, {\em Notre Dame J. Formal Logic} 47 (4) (2006), 515--524.

\bibitem{Herken} R. Herken. {\em The Universal Turing Machine: A Half-Century Survey}, 		
Oxford University Press, Oxford, 1992.

\bibitem{JS2012} C. G. Jockusch, Jr.,  P. Schupp. Generic computability, Turing degrees, and asymptotic density, 
{\em Journal of the London Mathematical Society} 85 (2) (2012), 472--490.

 \bibitem{Maurice-2010} M. Margenstern. Turing machines with two letters and two states,
 {\em Complex Systems} 19 (2010), 29--43.
 

\bibitem{YM2010} Yu. I.  Manin. {\em A Course in Mathematical Logic for Mathematicians}, Springer, Berlin, 1977; second edition, 2010.

\bibitem{YM2012} Yu. I.  Manin. Renormalisation and computation II: time cut-off and the Halting Problem,
{\em Math. Struct. in Comp. Science} 22 (2012), 729--751.



\bibitem{MM} M. Minsky.  Size and structure of universal Turing machines using Tag systems, in {\em Recursive Function Theory, Proc. Symp. Pure Mathematics},  AMS, Providence RI, 5, 1962,  229--238.

%\bibitem{NW1} N. Neary. D.  Woods. Small semi-weakly universal Turing machines,
%in J. Durand-Lose, M. Margenstern (eds.). {\em Machines, Computations and Universality 2007}, LNCS 4664, Springer, 2007, 303--315.



\bibitem{NW} N. Neary. D.  Woods. 
The complexity of small universal Turing machines,
  in  S. B. Cooper, B. Loewe,  A. Sorbi (eds.).    {\em Computability in Europe 2007},  LNCS 4497, CIE, Springer, 2007,  791--798.


\bibitem{YR} Y. Rogozhin. A universal Turing machine with 22 states and 2 symbols, {\em Romanian Journal of Information Science and Technology} 1 (3) (1998), 259--265.

\bibitem{CS} C. Shannon. A universal Turing machine with two internal states, in  {\em Automata Studies, Princeton},  Princeton University Press, NJ, 1956, 157--165.

%\bibitem{Smith} A. Smith. Wolfram's 2,3 Turing machine is universal, \url{https://www.wolframscience.com/prizes/tm23/TM23Proof.pdf}.


\bibitem{AT} A. Turing. On computable numbers, with an application to the Entscheidungsproblem, {\em Proceedings of the London Mathematical Society} 42 (2) (1936), 230--265.


\bibitem{ATc} A. Turing. On computable numbers, with an application to the Entscheidungsproblem: A correction,  {\em Proceedings of the London Mathematical Society} 2, 43 (1937), 544--546.


\bibitem{WN}  D.  Woods, N. Neary. 
The complexity of small universal Turing
machines: A survey, {\em Theor. Comput. Sci.} 410(4-5), (2009),  443--450.

%\bibitem{Watanabe}  S. Watanabe. 4-symbol 5-state universal Turing machine,
%{\em Information Processing Society of Japan Magazine} 13 (9) (1972), 588--592.


%\bibitem{Wolfram} Wolfram 2,3 Turing Machine, \url{http://demonstrations.wolfram.com/TheWolfram23TuringMachine/}.


\end{thebibliography}
\end{document}